\documentclass[runningheads]{llncs}
\usepackage[T1]{fontenc}
\usepackage{amsmath}        
\usepackage{amssymb}        
\usepackage{mathtools}      
\usepackage{enumitem}       
\usepackage{booktabs}       
\usepackage{array}          
\usepackage{caption}        
\usepackage{graphicx}       
\usepackage{subfigure}
\usepackage{wrapfig}        
\usepackage{xcolor}         
\usepackage[colorlinks=true,linkcolor=blue,citecolor=red,urlcolor=green]{hyperref}

\begin{document}
\title{SMTL: A Stratified Logic for Expressive Multi-Level Temporal Specifications}
%
%
\author{Ali Baheri\inst{1}\orcidID{0000-0002-5613-0192} \and
Peng Wei\inst{2}\orcidID{0000-0001-8492-5411}}
\authorrunning{A. Baheri et al.}
%
\institute{Rochester Institute of Technology, Rochester NY 14623, USA 
\email{akbeme@rit.edu}\\
\and
George Washington University, Washington DC, 20052, USA \email{pwei@gwu.edu}}
\maketitle              
\begin{abstract}
We present \textbf{\underline{S}}tratified \textbf{\underline{M}}etric \textbf{\underline{T}}emporal \textbf{\underline{L}}ogic (SMTL), a novel formalism for specifying and verifying properties of complex cyber-physical systems that exhibit behaviors across multiple temporal and abstraction scales. SMTL extends existing temporal logics by incorporating a stratification operator, enabling the association of temporal properties with specific abstraction levels. This allows for the natural expression of multi-scale requirements while maintaining formal reasoning about inter-level relationships. We formalize the syntax and semantics of SMTL, proving that it strictly subsumes metric temporal logic (MTL) and offers enhanced expressiveness by capturing properties unattainable in existing logics. Numerical simulations comparing agents operating under MTL and SMTL specifications show that SMTL enhances agent coordination and safety, reducing collision rates without substantial computational overhead or compromising path efficiency. These findings underscore SMTL's potential as a valuable tool for designing and verifying complex multi-agent systems operating across diverse temporal and abstraction scales.

\keywords{Stratified Metric Temporal Logic \and Multi-Scale Systems \and Formal Verification \and Multi-Agent Coordination \and Temporal Logic}
\end{abstract}

\vspace{-3 mm}
\section{Introduction}

Cyber-physical systems (CPS) exhibit behaviors across multiple temporal and abstraction scales, from millisecond-level control loops to hour-long mission objectives. While temporal logics have proven invaluable for specifying and verifying system properties \cite{bartocci2018specification}, existing formalisms such as metric temporal logic (MTL) and signal temporal logic (STL) struggle to efficiently capture and reason about such multi-scale behaviors. This limitation becomes particularly acute in systems like autonomous vehicles, where safety properties must be simultaneously maintained at the level of individual actuators, trajectory planning, and mission execution.

Existing temporal logics enforce a uniform treatment of time, requiring specifications to operate at the finest required temporal granularity \cite{eappenscaling}. Consider an autonomous drone delivery system: at the lowest level, motor control requires millisecond-precision timing; at the intermediate level, trajectory following operates on a seconds timescale; and at the highest level, delivery scheduling involves minutes or hours. Expressing all these requirements in MTL or STL necessitates reasoning at the millisecond scale throughout, leading to unnecessary complexity and computational overhead. The key challenge lies in the lack of formal mechanisms for stratifying temporal properties across different abstraction levels while maintaining sound logical relationships between these levels. Existing approaches typically handle multi-scale properties through informal decomposition or through separate specifications at different levels, losing the ability to reason formally about inter-level dependencies.

\noindent {\textbf{Paper Contributions.}} This paper introduces SMTL, an extension of traditional temporal logics incorporating a stratification operator to enable multi-scale temporal specification. The key contributions are:

\begin{itemize}

    \item We formalize the syntax and semantics of SMTL, proving that it strictly subsumes MTL and offers enhanced expressiveness by capturing properties unattainable in existing logics.

    \item We analyze the decidability and computational complexity of SMTL model checking, identifying decidable fragments with complexities comparable to MTL for bounded-time properties

    \item Through case studies and numerical simulations, we demonstrate SMTL's practical utility in specifying multi-scale requirements and improving verification efficiency, enhancing agent coordination and safety.
    
\end{itemize}
The practical utility of SMTL is demonstrated through case studies, where we show how SMTL specifications naturally capture requirements spanning multiple time scales while enabling more efficient verification compared to traditional approaches. We also present results on runtime monitoring, showing how SMTL's stratified nature enables efficient multi-rate monitoring strategies that adapt to the temporal requirements at each abstraction level.

\noindent {\textbf{Related Work.}} Various temporal logics have been developed to specify and reason about temporal properties in CPS and control systems. STL and MTL are prominent examples that extend LTL to handle real-valued signals and timing constraints \cite{momtaz2023monitoring}. Probabilistic STL (PrSTL) extends STL with probabilistic operators to reason about stochastic CPS \cite{yoo2015control}. Truncated linear temporal logic (LTL) is a variant of LTL interpreted over finite traces and has been used to specify complex rules for RL agents. Robustness TL (RobTL) is a logic for specifying and analyzing distances between CPS behaviors over a finite time horizon \cite{fainekos2006robustness}. The use of temporal logics like computation tree logic (CTL) and LTL provided a foundation for expressing desired behaviors in multi-agent systems \cite{fisher2005temporal,nilsson2016control,sahin2017provably,xu2016census,leahy2022fast,kantaros2020stylus}.

SMTL represents an advancement in the domain of formal methods for multi-agent systems, building upon the foundations of MTL to address the complexities of agent coordination and collision avoidance. MTL, introduced by Koymans, extends LTL by incorporating quantitative temporal constraints, enabling the specification of properties over real-time systems. While MTL has been effectively used for specifying timing constraints and verifying system behaviors, it lacks mechanisms to prioritize different types of constraints, which is crucial in multi-agent coordination where safety and goal achievement may conflict. SMTL is as a response to this limitation by introducing a stratification of temporal constraints, allowing agents to prioritize safety and coordination over individual objectives. The stratification concept aligns with the \emph{hierarchical} organization of specifications, where higher-priority constraints must be satisfied before considering lower-priority ones. This approach is particularly relevant in multi-agent systems, where agents operate concurrently and interactions can lead to emergent behaviors.

Previous efforts have explored the use of hierarchical and prioritized logics in multi-agent planning and verification. Lamport developed TLA as a framework for specifying and reasoning about concurrent systems, with a focus on simplifying temporal logic reasoning \cite{lamport1994temporal}. Separately, Alur and Henzinger developed real-time logics to enable quantitative reasoning about timing delays in real-time applications \cite{alur1993real}. Research in multi-agent coordination employs various formal methods, including temporal logics, to specify and verify interaction protocols and ensure safe coordination between agents. These methods focus on specifying conventions of social interaction and rules that govern agent behavior in multi-agent systems \cite{endriss2005temporal}. Furthermore, research in robot control synthesis has focused on using LTL to specify both safety requirements (describing how the robot should always behave) and liveness specifications (describing goals that must eventually be achieved) \cite{kress2018synthesis}.

Stratification has been well-established in database theory through Stratified Datalog, where it provides a way to handle negation and recursion by organizing rules into strata. This concept has been extended to temporal reasoning through formalisms like DatalogMTL, which incorporates metric temporal operators \cite{cucala2021stratified}. In recent years, the application of formal methods in multi-agent pathfinding has gained traction, with researchers leveraging satisfiability modulo theories (SMT) and other verification techniques to ensure safe navigation \cite{felli2021smt}. Sharon et al. developed conflict-based search (CBS) algorithms that decompose multi-agent pathfinding problems into individual agent plans while resolving conflicts iteratively \cite{sharon2015conflict}. These methods do not inherently incorporate stratified priorities in their formulations. The introduction of SMTL contributes to this landscape by providing a logical framework that \emph{inherently} supports priority stratification. By structuring temporal constraints into layers, SMTL facilitates the design of agents that can dynamically adjust their actions based on the priority of constraints, leading to improved coordination and safety in complex environments. Moreover, the integration of SMTL into multi-agent systems aligns with the broader trend of incorporating formal verification into agent-based modeling. The works of Belta et al. on formal methods for control synthesis have emphasized the importance of temporal logic in specifying desired behaviors, although the stratification aspect has remained less explored \cite{belta2019formal}. An additional dimension relevant to stratified logic is the use of multi-fidelity modeling and verification techniques \cite{shahrooei2024optimizing,baheri2023safety,baheri2023exploring,shahrooei2022falsification,beard2022safety}. Multi-fidelity approaches involve employing models of varying levels of detail and accuracy to balance the trade-off between computational efficiency and the precision of system analysis.



\noindent {\textbf{Paper Organization.}} The remainder of this paper is organized as follows. Section 2 provides necessary background on temporal logics and introduces the mathematical framework for stratified systems. Section 3 presents the formal syntax and semantics of SMTL, including practical examples demonstrating its expressiveness. Section 4 develops the theoretical results, proving SMTL's properties and complexity bounds. Section 5 presents experimental results comparing SMTL and MTL in multi-agent coordination scenarios. Finally, Section 6 concludes with a discussion of findings and future research directions.

\vspace{-3 mm}
\section{Preliminaries}

Before introducing SMTL, we review essential concepts from metric temporal logic and establish the mathematical framework for reasoning about stratified systems. Throughout this paper, we use the following notation:

\begin{table}[t]
\centering
\renewcommand{\arraystretch}{1.3}
\begin{tabular}{p{2.5cm}p{9cm}}
\toprule
\textbf{Symbol} & \textbf{Description} \\
\midrule
$\Sigma_k$ & Signal space at abstraction level $k$ \\
$\mathcal{L}(\varphi)$ & Language of formula $\varphi$ \\
$s \sqsubseteq_k s'$ & State $s$ refines state $s'$ at level $k$ \\
$\models_k$ & Satisfaction relation at abstraction level $k$ \\
$[a,b], (a,b)$, & Intervals with $a, b \in \mathbb{Q}_{\geq 0} \cup \{\infty\}$ \\
$[a,b), (a,b]$ & \\
$\Diamond\varphi$ & Shorthand for $\Diamond_{[0,\infty)}\varphi$ \\
\midrule
\multicolumn{2}{l}{\textbf{Complexity Measures}} \\
\midrule
$|\varphi|$ & Size of formula $\varphi$ \\
$\|\mathcal{M}\|$ & Size of transition system $\mathcal{M}$ \\
$K$ & Maximum abstraction level \\
$c_{\max}$ & Maximum constant in timing constraints \\
\bottomrule
\end{tabular}
\caption{Symbols and complexity measures used in the formalism.}
\label{tab:formalism_symbols}
\end{table}

\subsection{Metric Temporal Logic}

MTL extends propositional temporal logic with timing constraints on temporal operators. Let $AP$ be a set of atomic propositions and $\mathbb{I}$ be the set of non-empty intervals in $\mathbb{R}_{\geq 0}$ with endpoints in $\mathbb{Q}_{\geq 0} \cup \{\infty\}$.

\begin{definition}[MTL Syntax]
An MTL formula $\varphi$ is built according to the grammar:
\[\varphi ::= p \mid \neg\varphi \mid \varphi_1 \wedge \varphi_2 \mid \varphi_1 \mathcal{U}_I \varphi_2\]
where $p \in AP$ and $I \in \mathbb{I}$.
\end{definition}
Additional temporal operators are defined as abbreviations:
\begin{align*}
\Diamond_I\varphi &\triangleq \text{true} \mathcal{U}_I \varphi & \text{(eventually)} \\
\Box_I\varphi &\triangleq \neg\Diamond_I\neg\varphi & \text{(always)} \\
\varphi_1 \mathcal{R}_I \varphi_2 &\triangleq \neg(\neg\varphi_1 \mathcal{U}_I \neg\varphi_2) & \text{(release)}
\end{align*}

\begin{definition}[Timed State Sequence]
A timed state sequence is a pair $\rho = (\sigma, \tau)$ where:
\begin{itemize}
\item $\sigma = \sigma_0\sigma_1\sigma_2\ldots$ is an infinite sequence of states $\sigma_i \in 2^{AP}$
\item $\tau = \tau_0\tau_1\tau_2\ldots$ is an infinite sequence of time stamps $\tau_i \in \mathbb{R}_{\geq 0}$, with $\tau_0 = 0$ and $\tau_i < \tau_{i+1}$ for all $i \geq 0$
\end{itemize}
\end{definition}

\begin{definition}[MTL Semantics]
For a timed state sequence $\rho = (\sigma, \tau)$ and position $i \geq 0$, the satisfaction relation $\models$ is defined inductively:
\begin{align*}
(\rho,i) \models p &\iff p \in \sigma_i \\
(\rho,i) \models \neg\varphi &\iff (\rho,i) \not\models \varphi \\
(\rho,i) \models \varphi_1 \wedge \varphi_2 &\iff (\rho,i) \models \varphi_1 \text{ and } (\rho,i) \models \varphi_2 \\
(\rho,i) \models \varphi_1 \mathcal{U}_I \varphi_2 &\iff \exists j \geq i: (\rho,j) \models \varphi_2, \tau_j - \tau_i \in I, \\
&\quad \text{and } \forall k, i \leq k < j: (\rho,k) \models \varphi_1
\end{align*}
\end{definition}

\subsection{Stratified Systems and Abstraction Hierarchies}

We now introduce the mathematical framework for reasoning about systems at multiple abstraction levels.

\begin{definition}[Abstraction Function]
An abstraction function $\alpha: S \rightarrow S'$ maps concrete states to abstract states, where $S$ and $S'$ are state spaces. An abstraction hierarchy is a sequence of state spaces and abstraction functions $\{(S_k, \alpha_k)\}_{k=1}^K$ where:
\begin{itemize}
\item $S_k$ is the state space at level $k$
\item $\alpha_k: S_{k-1} \rightarrow S_k$ maps states from level $k-1$ to level $k$
\end{itemize}
\end{definition}

\begin{definition}[Abstraction Properties]
An abstraction hierarchy satisfies:
\begin{enumerate}
\item \textbf{Monotonicity}: For $i < j$, $\alpha_j \circ \alpha_i = \alpha_j$
\item \textbf{Preservation}: For any property $\varphi$ preserved by $\alpha_k$, if $s \models \varphi$ then $\alpha_k(s) \models \varphi$
\item \textbf{Refinement}: For states $s_1, s_2$, if $\alpha_k(s_1) = \alpha_k(s_2)$, then $s_1$ and $s_2$ are equivalent at level $k$
\end{enumerate}
\end{definition}

\begin{definition}[Temporal Resolution]
Each abstraction level $k$ has an associated temporal resolution $\rho_k \in \mathbb{R}_{>0}$, representing the minimum granularity of time distinguishable at that level. The temporal resolution hierarchy satisfies:
\[\forall k < l: \rho_k < \rho_l\]
\end{definition}
\vspace{-3 mm}
\section{Syntax and Semantics of SMTL}

\subsection{Syntax}

\begin{definition}[SMTL Syntax]
Let $AP$ be a set of atomic propositions and $\mathbb{I}$ be the set of non-empty intervals in $\mathbb{R}_{\geq 0}$ with endpoints in $\mathbb{Q}_{\geq 0} \cup \{\infty\}$. An SMTL formula $\varphi$ is built according to the grammar:
\[\varphi ::= p \mid \neg\varphi \mid \varphi_1 \wedge \varphi_2 \mid \varphi_1 \mathcal{U}_I \varphi_2 \mid L_k\varphi\]
where $p \in AP$ is an atomic proposition, $I \in \mathbb{I}$ is a time interval, $k \in \mathbb{N}$ is an abstraction level, and $L_k$ is the stratification operator for level $k$.
\end{definition}
As in MTL, we define additional temporal operators as abbreviations:
\begin{align*}
\Diamond_I\varphi &\triangleq \text{true} \mathcal{U}_I \varphi \\
\Box_I\varphi &\triangleq \neg\Diamond_I\neg\varphi \\
\varphi_1 \mathcal{R}_I \varphi_2 &\triangleq \neg(\neg\varphi_1 \mathcal{U}_I \neg\varphi_2)
\end{align*}
The stratification operator $L_k$ can be nested, allowing for relationships between abstraction levels:
\[L_i(L_j\varphi) \text{ where } i,j \in \mathbb{N}\]

\begin{definition}[Well-Formed SMTL Formula]
An SMTL formula $\varphi$ is well-formed if for any nested stratification operators $L_i$ and $L_j$ appearing in $\varphi$, where $L_i$ is within the scope of $L_j$, we have $i \leq j$.
\end{definition}

\subsection{Semantics}

The semantics of SMTL is defined over timed state sequences with multiple abstraction levels.

\begin{definition}[Stratified Timed State Sequence]
A stratified timed state sequence is a tuple $\rho = (\{\sigma_k\}_{k=1}^K, \tau)$ where:
\begin{itemize}
\item $\sigma_k = \sigma_{k,0}\sigma_{k,1}\sigma_{k,2}\ldots$ is an infinite sequence of states at level $k$, with $\sigma_{k,i} \in 2^{AP}$
\item $\tau = \tau_0\tau_1\tau_2\ldots$ is an infinite sequence of time stamps with $\tau_0 = 0$ and $\tau_i < \tau_{i+1}$
\item For each $k$, consecutive states in $\sigma_k$ are separated by at least $\rho_k$ time units
\end{itemize}
\end{definition}

\begin{definition}[Abstraction Consistency]
A stratified timed state sequence $\rho$ is abstraction consistent if for all levels $i < j$ and all positions $n$:
\[\alpha_j(\sigma_{i,n}) = \sigma_{j,n}\]
where $\alpha_j$ is the abstraction function from level $i$ to level $j$.
\end{definition}

\begin{definition}[SMTL Semantics]
For a stratified timed state sequence $\rho = (\{\sigma_k\}_{k=1}^K, \tau)$, position $i \geq 0$, and abstraction level $m$, the satisfaction relation $\models_m$ is defined inductively:

\begin{align*}
(\rho,i) \models_m p &\iff p \in \sigma_{m,i} \\
(\rho,i) \models_m \neg\varphi &\iff (\rho,i) \not\models_m \varphi \\
(\rho,i) \models_m \varphi_1 \wedge \varphi_2 &\iff (\rho,i) \models_m \varphi_1 \text{ and } (\rho,i) \models_m \varphi_2 \\
(\rho,i) \models_m \varphi_1 \mathcal{U}_I \varphi_2 &\iff \exists j \geq i: (\rho,j) \models_m \varphi_2, \tau_j - \tau_i \in I, \\
&\quad \text{and } \forall k, i \leq k < j: (\rho,k) \models_m \varphi_1 \\
(\rho,i) \models_m L_k\varphi &\iff (\rho,i) \models_k \varphi \text{ and } k \geq m
\end{align*}
\end{definition}
The semantics of the stratification operator $L_k$ ensures that:
\begin{enumerate}
\item Properties at level $k$ are evaluated using the state sequence $\sigma_k$
\item The temporal resolution at level $k$ is respected ($\rho_k$)
\item Abstraction consistency is maintained between levels
\end{enumerate}

\subsection{Practical Implications of SMTL Expressiveness}

Here, we provide two examples that demonstrate the enhanced expressiveness of SMTL over MTL.

\noindent {\textbf{Example 1: Multi-Scale System Specification.}} Consider an autonomous vehicle navigation system operating across multiple time scales. Existing MTL specifications struggle to capture requirements that span different temporal granularities. SMTL naturally expresses such multi-scale requirements through a single coherent formula:

\[\begin{aligned}
\phi_{\text{nav}} = &L_1\Box_{[0,0.01]}(\|a_{\text{actual}} - a_{\text{commanded}}\| \leq \epsilon) \land {} \\
&L_2\Box_{[0,1]}(\text{speed} > v_{\text{min}} \rightarrow \Diamond_{[0,0.5]}\text{lane\_centered}) \land {} \\
&L_3\Box_{[0,60]}(\text{destination\_reached} \rightarrow \Box_{[0,5]}\text{safely\_parked})
\end{aligned}\]
This formula simultaneously captures millisecond-level actuator control ($L_1$), second-level trajectory tracking ($L_2$), and minute-level mission objectives ($L_3$). The stratification operator cleanly separates these concerns while maintaining their logical relationships.

\noindent {\textbf{Example 2: Abstraction in Software Architecture.}} The expressiveness of SMTL aligns with layered software architectures. Consider a robotic system with a three-tier architecture:

\[\begin{aligned}
\phi_{\text{arch}} = &L_1(\text{functional\_layer\_spec}) \land {} \\
&L_2(\text{executive\_layer\_spec}) \land {} \\
&L_3(\text{planning\_layer\_spec})
\end{aligned}\]
where each layer's specification might involve temporal properties. The crucial advantage lies in SMTL's ability to express inter-layer dependencies:

\[\begin{aligned}
\phi_{\text{deps}} = L_2\Box_{[0,\infty)}&(\text{plan\_update} \rightarrow {} \\
&L_1\Diamond_{[0,\delta]}(\text{exec\_acknowledge} \land {} \\
&\Diamond_{[0,\tau]}\text{functional\_update}))
\end{aligned}\]
This hierarchical specification captures both the independence of different architectural layers and their temporal coupling, a feat impossible in standard MTL.

\section{Theoretical Analysis}

\begin{theorem}[Stratification Soundness]
For any SMTL formula $\varphi$ and abstraction levels $i < j$:
\[(\rho,n) \models_i L_j\varphi \implies (\rho,n) \models_j \varphi\]
\end{theorem}

\begin{proof}
Follows directly from the semantics of $L_k$ and abstraction consistency. When $(\rho,n) \models_i L_j\varphi$, the formula $\varphi$ is evaluated at level $j$, using $\sigma_j$, which by abstraction consistency is a valid abstraction of the states at level $i$.
\end{proof}

\begin{theorem}[Temporal Resolution Hierarchy]
For any SMTL formula $\varphi$ and levels $i < j$:
\[\text{If } (\rho,n) \models_i \Diamond_{[0,t]}\varphi \text{ then } t \geq \rho_i\]
\end{theorem}

\begin{proof}
By the definition of stratified timed state sequences, states at level $i$ must be separated by at least $\rho_i$ time units. Therefore, any temporal property at level $i$ must span at least this duration.
\end{proof}

\subsection{Expressiveness of SMTL}

We begin by establishing the foundational properties of abstraction functions that underpin the stratified nature of our logic. An abstraction function $\alpha_k: \Sigma \to \Sigma_k$ maps a concrete signal $\sigma \in \Sigma$ to an abstract signal in abstraction level $k$, where $\Sigma_k$ represents the signal space at level $k$. These functions form a hierarchy satisfying crucial properties: monotonicity ensures that for $i < j$, $\alpha_i \circ \alpha_j = \alpha_i$; consistency guarantees that if $\phi$ holds for $\alpha_i(\sigma)$, then there exists some higher level $j > i$ where $\phi$ holds for $\alpha_j(\sigma)$; and the refinement property ensures that distinct signals remain distinguishable at some abstraction level.

\begin{theorem}[Strict Expressiveness]
SMTL is strictly more expressive than MTL. This relationship manifests both in SMTL's ability to express all MTL properties and in its capacity to express properties that are inexpressible in MTL.
\end{theorem}

\begin{proof}
The proof proceeds in two parts, first establishing that SMTL subsumes MTL, then demonstrating that SMTL can express properties beyond MTL's capabilities. To show that SMTL subsumes MTL, we construct a translation function $\tau: \text{MTL} \to \text{SMTL}$ defined inductively on the structure of MTL formulas:

\begin{align*}
    \tau(p) &= p \text{ for } p \in AP \\
    \tau(\neg\phi) &= \neg\tau(\phi) \\
    \tau(\phi_1 \wedge \phi_2) &= \tau(\phi_1) \wedge \tau(\phi_2) \\
    \tau(\phi_1 \mathcal{U}_{[a,b]} \phi_2) &= \tau(\phi_1) \mathcal{U}_{[a,b]} \tau(\phi_2)
\end{align*}
Through structural induction, we prove that for any signal $\sigma$ and time point $t$, the satisfaction relation preserves equivalence: $(\sigma,t) \models_{\text{MTL}} \phi \iff (\sigma,t) \models_{\text{SMTL}} \tau(\phi)$. The base case for atomic propositions follows directly from the definition. For the inductive step, we consider each operator, with particular attention to the metric until operator:

\begin{align*}
    (\sigma,t) &\models_{\text{MTL}} \phi_1 \mathcal{U}_{[a,b]} \phi_2 \\
    &\iff \exists t' \in [t+a,t+b]: (\sigma,t') \models_{\text{MTL}} \phi_2 \land \forall t'' \in [t,t'): (\sigma,t'') \models_{\text{MTL}} \phi_1 \\
    &\iff \exists t' \in [t+a,t+b]: (\sigma,t') \models_{\text{SMTL}} \tau(\phi_2) \land \forall t'' \in [t,t'): (\sigma,t'') \models_{\text{SMTL}} \tau(\phi_1) \\
    &\iff (\sigma,t) \models_{\text{SMTL}} \tau(\phi_1) \mathcal{U}_{[a,b]} \tau(\phi_2)
\end{align*}
To establish SMTL's strictly greater expressiveness, we construct an SMTL formula that cannot be expressed in MTL. Consider the formula:

\[\psi = L_1\Box_{[0,1]}(p) \wedge L_2\Diamond_{[0,2]}(\neg p)\]
We prove its MTL-inexpressibility through a contradiction argument. Let $\Sigma = \{0,1\}^{\mathbb{R}_{\geq 0}}$ be the set of boolean signals. Define two signals $\sigma_1$ and $\sigma_2$ that differ only at a single point:

\[\sigma_1(t) = \begin{cases}
    1 & t \in [0,1] \\
    0 & \text{otherwise}
\end{cases}\]

\[\sigma_2(t) = \begin{cases}
    1 & t \in [0,1] \setminus \{0.5\} \\
    0 & \text{otherwise}
\end{cases}\]
Define abstraction functions where $\alpha_1$ smooths out isolated points while $\alpha_2$ maintains the original signal\footnote{In this context, $\alpha_2 = \text{id}$ means ($\alpha_2$) is the identity function. The identity function is being used at abstraction level 2 to represent "no abstraction" or "keep all details". This contrasts with $\alpha_1$ which does perform abstraction.}:

\[\alpha_1(\sigma)(t) = \begin{cases}
    1 & \text{if } \exists \delta > 0: \forall t' \in (t-\delta,t+\delta), \sigma(t') = 1 \\
    0 & \text{otherwise}
\end{cases}\]

\[\alpha_2 = \text{id}\]
Under these definitions, $(\sigma_1,0) \models_{\text{SMTL}} \psi$ while $(\sigma_2,0) \not\models_{\text{SMTL}} \psi$. However, any MTL formula $\theta$ must evaluate equivalently on both signals: $(\sigma_1,0) \models_{\text{MTL}} \theta \iff (\sigma_2,0) \models_{\text{MTL}} \theta$, as MTL cannot distinguish signals differing at isolated points. This contradiction establishes that no MTL formula can express the property captured by $\psi$.
\end{proof}

\begin{corollary}[Stratification Power]
The stratification operator $L_k$ introduces fundamental expressive capabilities that transcend traditional temporal operators, enabling reasoning about properties at different abstraction levels simultaneously.
\end{corollary}

\begin{corollary}[Multi-Scale Properties]
SMTL provides native support for specifying and reasoning about temporal properties across multiple time scales, a capability that proves essential in complex systems where behaviors manifest at different granularities of time and abstraction.
\end{corollary}

\begin{theorem}[SMTL Model Checking]
Given an SMTL formula $\phi$, a finite-state transition system $\mathcal{M}$, and maximum abstraction level $K$ in $\phi$, the model checking problem $\mathcal{M} \models \phi$ is EXPTIME-complete for bounded-time formulas and 2EXPTIME-complete for unbounded-time formulas. Moreover, the complexity grows polynomially with the number of abstraction levels $K$.
\end{theorem}

\begin{proof}
We begin by establishing the theoretical framework for model checking SMTL formulas through a series of constructions that build upon the stratified nature of the logic. Let $\mathcal{M} = (S, S_0, R, L)$ be our finite-state transition system, where $S$ represents the state space, $S_0$ the initial states, $R$ the transition relation, and $L$ the labeling function. At each abstraction level $k \leq K$, we define an abstract transition system $\alpha_k(\mathcal{M}) = (S_k, S_{0,k}, R_k, L_k)$. The states $S_k$ are obtained through the abstraction function $\alpha_k$ applied to $S$, with corresponding abstractions for initial states, transitions, and labels. The critical property of these abstractions is the preservation of the stratification hierarchy: for any levels $i < j$, we have $\alpha_i \circ \alpha_j = \alpha_i$.

The model checking algorithm proceeds through the construction of a region automaton $\mathcal{R}(\mathcal{M}, \phi)$. The states of this automaton are tuples $(s, \nu, \Phi)$, where $s$ represents a state of $\mathcal{M}$, $\nu$ captures clock valuations, and $\Phi$ maintains the set of subformulas requiring verification. Clock constraints are derived from the temporal bounds appearing in $\phi$, with $c_{\max}$ denoting the maximum such bound. For temporal subformulas at each abstraction level $k$, we introduce dedicated clocks $x_{\psi}$ for each subformula $L_k\psi$. The size of the resulting region automaton is bounded by $|S| \cdot (c_{\max}+1)^{|\text{clocks}|} \cdot 2^{|\phi|} \cdot K$. This bound reflects the product of the original state space size, the clock regions, the subformula combinations, and the abstraction levels. The verification process traverses the abstraction hierarchy from bottom to top. At each level $k$, we construct the corresponding abstract system $\alpha_k(\mathcal{M})$ and its region automaton $\mathcal{R}_k$. The satisfaction of formulas at level $k$ is evaluated within $\mathcal{R}_k$, with results propagating upward through the abstraction hierarchy. The consistency between levels is maintained through the abstraction functions.

For bounded-time formulas, the complexity analysis reveals an EXPTIME upper bound. This follows from the size of the region automaton and the complexity of standard CTL model checking techniques applied at each abstraction level. The total complexity is dominated by the term $O(K \cdot |S| \cdot (c_{\max}+1)^{|\text{clocks}|} \cdot 2^{|\phi|})$. For unbounded-time formulas, the necessity to handle infinite paths introduces an additional exponential factor, yielding 2EXPTIME complexity. To establish EXPTIME-hardness for the bounded case, we present a reduction from the acceptance problem for linearly bounded alternating Turing machines. Given such a machine $M$, we construct an SMTL formula $\phi_M$ that captures both the validity of configurations and the transition structure of $M$:

\[\phi_M = L_1\Box_{[0,T]}(\text{valid\_config}) \land \bigwedge_{q \in Q} L_2\Diamond_{[0,T]}(q \rightarrow \text{next\_config}(q))\]
Here, $T$ bounds the computation steps of $M$, while valid\_config and next\_config encode the machine's configuration space and transition function respectively. The construction ensures that $M$ accepts input $w$ if and only if $\mathcal{M}_w \models \phi_M$, where $\mathcal{M}_w$ encodes the input. The polynomial growth with respect to $K$ follows from the additive nature of abstraction levels. Each new level introduces one additional region automaton of size $O(|\mathcal{R}|)$ and requires consistency checks with the previous level, contributing $O(|\mathcal{R}|^2)$ to the complexity. Thus, the total complexity grows as $O(K \cdot |\mathcal{R}|^2)$. For practical applications, when dealing with bounded formulas having a fixed number of clocks and constants, the complexity reduces to $O(|S| \cdot K \cdot 2^{|\phi|})$. This bound suggests the feasibility of verification for realistic systems when employing incremental verification techniques that localize the impact of adding new abstraction levels.
\end{proof}

\begin{corollary}[Incremental Verification]
The addition of a new abstraction level to a verified system requires only the verification of properties specific to the new level and consistency checks with the adjacent level, enabling efficient incremental verification procedures.
\end{corollary}
\vspace{-3 mm}
\section{Numerical Results}

To evaluate the effectiveness of SMTL in enhancing multi-agent coordination and safety, we conducted a series of simulations comparing the performance of agents operating under SMTL and MTL specifications. The experiments were designed to assess how the introduction of stratification in temporal logic influences agent behavior in terms of collision avoidance and path efficiency. In our experimental setup, agents navigated grid worlds of dimensions ranging from $5 \times 5$ to $150 \times 150$. Each grid represented a discrete environment in which agents moved from randomly assigned starting positions to specified target locations. Obstacles were randomly placed within the grids. Agents were programmed to move one unit per time step in any of the four cardinal directions, provided the destination cell was within grid boundaries and not occupied by an obstacle or another agent. The primary objective for each agent was to reach its goal while avoiding collisions with other agents and obstacles.

Under the MTL framework, agents were governed by temporal logic specifications that dictated their movement towards goals without consideration of other agents. The MTL specification used can be mathematically expressed as:

$$
\phi_{\mathrm{MTL}}=\square_{[0, T]} \text { (reachGoal} \rightarrow \diamond_{[0, T]} \text {atGoal})
$$
This formula specifies that globally, from time 0 to $T$, if an agent intends to reach its goal (reachGoal), then eventually, within time $T$, it must be at the goal position (atGoal). In contrast, agents operating under SMTL incorporated stratified reasoning to anticipate and avoid potential conflicts with other agents. The SMTL specification extended the MTL formula to include collision avoidance at a higher level of reasoning:

$$
\phi_{\mathrm{SMTL}}=L_1\left(\phi_{\mathrm{MTL}} \wedge \square_{[0, T]}\left(\forall j \neq i, \neg\left(\operatorname{pos}_i=\operatorname{pos}_j\right)\right)\right)
$$
Here, $L_1$ represents the stratification level where agents consider not only their goal-reaching objectives but also the positions of other agents to prevent collisions. The term $\square_{[0, T]}\left(\forall j \neq i, \neg\left(\operatorname{pos}_i=\operatorname{pos}_j\right)\right)$ ensures that at all times, agents avoid occupying the same position as any other agent. In this case study, the stratified nature matches the temporal stratification in SMTL through the separation of agent behaviors and constraints according to their temporal characteristics:

\begin{itemize}
    \item \textbf{Short-Term Temporal Constraints (Fine-Grained Temporal Level):} Agents using SMTL enforce collision avoidance at every time step. This is a temporal property that requires agents to ensure that they do not occupy the same position as any other agent. In SMTL, this corresponds to a higher-priority constraint specified at a lower stratification level (e.g., $L_1$), which operates at a finer temporal granularity. The logic enforces that safety properties hold at all times or within very short time intervals.
    \item \textbf{Long-Term Temporal Objectives (Coarse-Grained Temporal Level):} Agents aim to reach their designated goals, which is a temporal property considered over a longer time horizon. In SMTL, this is captured at a higher stratification level (e.g., $L_2$ or $L_3$), where the temporal properties involve longer intervals or eventualities.
\end{itemize}

\begin{figure}[t]
    \centering
    \subfigure[a]{
        \includegraphics[width=0.45\textwidth]{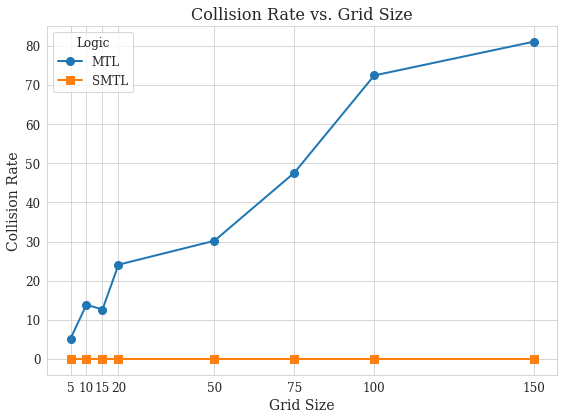}
    }
    \subfigure[b]{
        \includegraphics[width=0.45\textwidth]{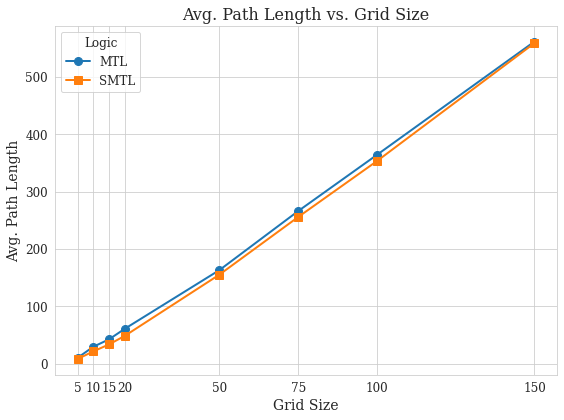}
    }
    \vspace{0.5cm} 
    \subfigure[c]{
        \includegraphics[width=0.45\textwidth]{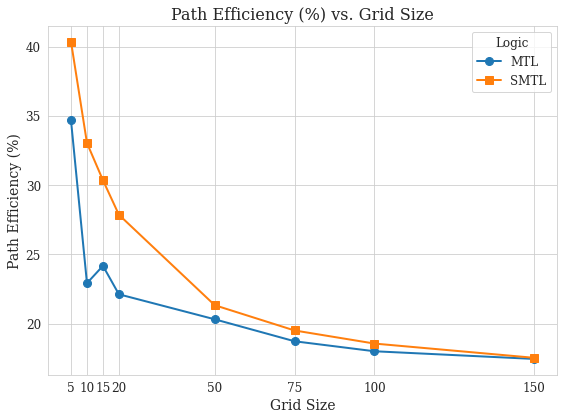}
    }
    \subfigure[d]{
        \includegraphics[width=0.45\textwidth]{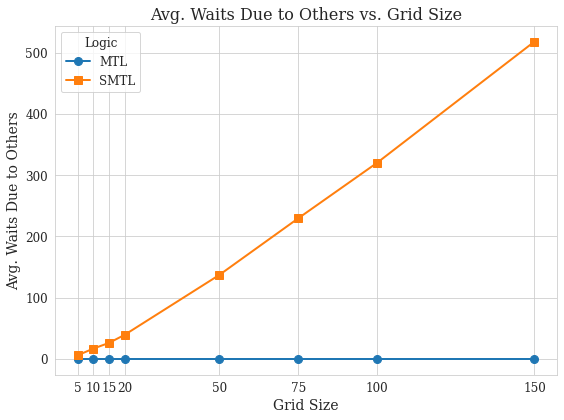}
    }
    \caption{Comparison of agent performance metrics under MTL and SMTL across different grid sizes: (a) collision rate, (b) average path length, (c) path efficiency, and (d) average waits due to others.}
    \label{fig:results}
\end{figure}

The experimental results presented in Figure \ref{fig:results} provide a comparison of agent performance under MTL and SMTL. For each grid size, the number of agents corresponds to the grid dimension, ensuring a consistent agent density across different environment scales. The key performance metrics evaluated include collision rate, average path length, path efficiency, and average waits due to others. A key observation from the data is the contrast in collision rates between the MTL and SMTL agents. MTL agents show a non-zero collision rate across all grid sizes, with the collision rate generally increasing as the grid size and the number of agents grow. Specifically, MTL agents experience an average collision rate of 5.20 collisions per agent in the $5 \times 5$ grid, which escalates to over 81 collisions per agent in the $150 \times 150$ grid. This trend underscores the limitations of MTL agents in avoiding collisions, particularly in larger environments where agent interactions are more frequent.

In contrast, SMTL agents maintain a collision rate of zero across all grid sizes, demonstrating their effectiveness in collision avoidance regardless of the environment's scale. This result highlights the efficacy of the stratification approach in SMTL, which allows agents to prioritize safety and coordination over individual objectives. The ability of SMTL agents to completely avoid collisions, even in densely populated and extensive grids, signifies a substantial improvement over traditional MTL agents and validates the practical utility of SMTL in multi-agent systems. Analyzing the average path length reveals that SMTL agents often have shorter or comparable path lengths compared to MTL agents, especially in smaller grids. For example, on the $5 \times 5$ grid, SMTL agents achieve an average path length of 8.40 units, whereas MTL agents average 10.47 units. This suggests that SMTL agents are capable of reaching their goals more efficiently in less congested environments. However, as the grid size increases, the difference in average path length between SMTL and MTL agents diminishes. In the $150 \times 150$ grid, SMTL agents have an average path length of approximately 558.57 units compared to 561.24 units for MTL agents. This convergence indicates that in larger environments, both agent types require similar path lengths to reach their goals, likely due to the proportional increase in distances that need to be traversed. The path efficiency metric, calculated as the ratio of the shortest possible path length to the actual path length, provides insight into the agents' navigational effectiveness relative to the optimal path. SMTL agents consistently show higher path efficiency percentages than MTL agents across all grid sizes. In the $5 \times 5$ grid, SMTL agents attain a path efficiency of approximately 40.37\%, surpassing the 34.68\% efficiency of MTL agents. Although both agent types experience a decline in path efficiency as the grid size increases—reflecting the greater complexity and potential detours in larger environments—SMTL agents maintain a slight advantage. This efficiency suggests that SMTL agents effectively balance the need for collision avoidance with the pursuit of efficient paths. A notable distinction between the two agent types is observed in the average waits due to others. SMTL agents exhibit a substantial number of waits, which increases with grid size—from an average of 5.93 waits in the $5 \times 5$ grid to over 517 waits in the $150 \times 150$ grid. This behavior reflects the SMTL agents' strategy of waiting to avoid potential collisions, emphasizing their prioritization of safety and adherence to collision avoidance constraints. In contrast, MTL agents have zero waits across all grid sizes, indicating that they proceed toward their goals without regard for other agents, even if it results in collisions. The willingness of SMTL agents to wait highlights their commitment to coordination and safe operation within shared environments.
\begin{wrapfigure}{r}{0.35\textwidth}
  \centering
  \includegraphics[width=0.35\textwidth]{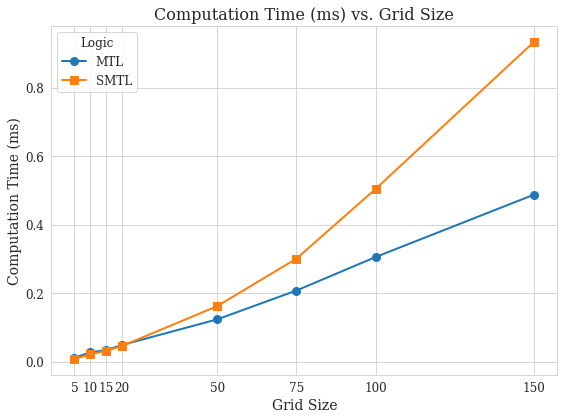} 
  \caption{Comparison of computation time per step between two agents}
  \label{fig:example}
\end{wrapfigure}
Computation time per step per agent is another metric where SMTL agents demonstrate competitive performance. While SMTL agents incur slightly higher computation times compared to MTL agents, the difference is marginal. For example, in the $150 \times 150$ grid, SMTL agents have an average computation time of approximately 0.935 milliseconds per step, compared to 0.489 milliseconds for MTL agents. This modest increase can be attributed to the additional processing required for collision avoidance and coordination with other agents. The results indicate that the coordination achieved by SMTL agents do not come at the expense of prohibitive computational costs.

\vspace{-3 mm}
\section{Conclusion}

In this paper, we proposed SMTL as a formalism for specifying and verifying properties of complex cyber-physical systems exhibiting behaviors across multiple temporal and abstraction scales. By incorporating a stratification operator into traditional temporal logics, SMTL enables the association of temporal properties with specific abstraction levels, facilitating the natural expression of multi-scale requirements while preserving formal reasoning about inter-level relationships. We formalized the syntax and semantics of SMTL, proving that it strictly subsumes MTL and offers enhanced expressiveness by capturing properties unattainable in existing logics. Through numerical simulations, we demonstrated the practical utility of SMTL in enhancing agent coordination and safety. The simulations comparing agents operating under MTL and SMTL specifications revealed that SMTL significantly reduces collision rates without incurring substantial computational overhead or compromising path efficiency. Specifically, SMTL agents maintained a zero collision rate in various grid sizes, highlighting the effectiveness of stratified reasoning in avoiding conflicts and improving overall system reliability. Our findings underscore the potential of SMTL as a valuable tool for designing and verifying complex multi-agent systems operating across diverse temporal and abstraction scales. 

\noindent {\textbf{Future Work.}} Future research may explore the integration of SMTL with automated synthesis tools to generate controllers that guarantee compliance with multi-scale specifications. Furthermore, investigating the application of SMTL in other domains, such as distributed sensor networks or human-robot interaction, could further validate its effectiveness. Extending SMTL to incorporate probabilistic reasoning or learning-based components may also enhance its applicability to systems with inherent uncertainties. Finally, the development of optimized algorithms for the checking of the SMTL model can facilitate its adoption in large-scale industrial applications.

\bibliographystyle{splncs04}
\bibliography{ref,ref2}

\end{document}